\documentclass{article}
\usepackage[utf8]{inputenc}
\usepackage{amsthm}
\usepackage{amsmath}
\usepackage{amsfonts}
\usepackage{amssymb} 
\usepackage{latexsym}
\usepackage{tikz}
\newtheorem{theorem}{Theorem}
\newtheorem{lemma}[theorem]{Lemma}
\newtheorem{corollary}[theorem]{Corollary}
\usepackage{xcolor}

\newcommand{\claimproof}{\noindent\emph{Proof of claim.} }
\newcommand{\smallqed}{{\tiny ($\Box$)}}

\newtheorem{claim}[theorem]{Claim}
 \usepackage{hyperref}

	{\noindent {}{#1}{}}{ This proves Claim~\arabic{claim}.\hfill
    $\Diamond$\vspace{1em}}
    
\usepackage{graphicx}
\graphicspath{{figures/}}

\title{Locating Dominating Sets in local tournaments}
\author{Thomas Bellitto\footnote{Faculty of Mathematics, Informatics and Mechanics, University of Warsaw, Poland. \linebreak This author is supported by the European Research Council (ERC)
under the European Union's Horizon 2020 research and innovation programme Grant Agreement 714704.}, Caroline Brosse\footnote{Universit\'e Clermont Auvergne, Clermont Auvergne INP, CNRS, Mines Saint-Etienne, LIMOS, F-63000 Clermont-Ferrand, France. This author is supported by the French government IDEX-ISITE initiative 16-IDEX-0001 (CAP 20-25).}, Benjamin L\'ev\^eque\footnote{CNRS, Laboratoire G-SCOP, Grenoble, France. This author is partially supported by the French ANR Projects GATO (ANR-16-CE40-0009-01)}, Aline Parreau\footnote{LIRIS, Université Claude Bernard Lyon 1, CNRS, France}}
\date{\today}

\begin{document}

\maketitle


\begin{abstract}
A dominating set in a directed graph is a set of vertices $S$ such that all the vertices that do not belong to $S$ have an in-neighbour in $S$. A locating set $S$ is a set of vertices such that all the vertices that do not belong to $S$ are characterized uniquely by the in-neighbours they have in $S$, \textit{i.e.} for every two vertices $u$ and $v$ that are not in $S$, there exists a vertex $s\in S$ that dominates exactly one of them. The size of a smallest set of a directed graph $D$ which is both locating and dominating is denoted by $\gamma^{LD}(D)$. Foucaud, Heydarshahi and Parreau proved that any twin-free digraph $D$ satisfies $\gamma^{LD}(D)\leq \frac{4n} 5 +1$ but conjectured that this bound can be lowered to $\frac{2n} 3$. The conjecture is still open. They also proved that if $D$ is a tournament, \textit{i.e.} a directed graph where there is one arc between every pair of vertices, then $\gamma^{LD}(D)\leq \lceil \frac{n}{2}\rceil$.

The main result of this paper is the generalization of this bound to connected local tournaments, \textit{i.e.} connected digraphs where the in- and out-neighbourhoods of every vertex induce a tournament. We also prove $\gamma^{LD}(D)\leq \frac{2n} 3$ for all quasi-twin-free digraphs $D$ that admit a supervising vertex (a vertex from which any vertex is reachable). This class of digraphs generalizes twin-free acyclic graphs, the most general class for which this bound was known.

\end{abstract}


\section{Introduction}

In this paper we consider loopless and finite digraphs. Our terminology is consistent with \cite{bookjbjgutin}. Especially, we refer the reader to this book for further information about the classes of digraphs we consider. 
The \emph{order} of a digraph is its number of vertices. A digraph $D$ is \emph{simple} if there is at most one arc between two vertices $x$ and $y$.
A digraph $D$ is called \emph{connected}, if for every vertices $x$ and $y$ of $D$ there exists a (non necessarily directed) path from $x$ to $y$.
A digraph $D$ is called \emph{strongly connected} (or \emph{strong} for short), if for every vertices $x$ and $y$ of $D$, there exists a directed path from $x$ to $y$ and a directed path from $y$ to $x$.

Let $x$, $y$ be distinct vertices of a digraph $D$. If there is an arc from $x$ to $y$, we say that $x$
\emph{dominates} $y$ and denote it by $x \to y$.
We say that $y$ is an \emph{out-neighbour} of $x$, and that $x$ is an \emph{in-neighbour} of $y$.
The \emph{open in-neighbourhood} and the \emph{open out-neighbourhood} of a vertex $v$,
denoted by $N^{-}(v)$ and $N^{+}(v)$ respectively, are the set of in-neighbours of $v$ and the set of out-neighbours
of $v$, respectively.  A \emph{source} is a vertex with no in-neighbours, and a \emph{sink} is a
vertex with no out-neighbours.
We use the notations $d^{-}(v)=|N^{-}(v)|$ and $d^{+}(v)=|N^{+}(v)|$ to denote the in-degree and the out-degree of $v$.
Further, the \emph{closed in-neigbourhood} of $v$ is $N^{-}[v] = N^{-}(v) \cup \{v\}$ and the \emph{closed out-neigbourhood} of $v$ is $N^{+}[v] = N^{+}(v) \cup \{v\}$.
Two vertices are called \emph{twins} if $N^{-}(x)=N^{-}(y)$ or if $N^{-}[x]=N^{-}[y]$, and \emph{quasi twins} if $N^{-}(x) = N^{-}[y]$ or $N^{-}(y) = N^{-}[x]$.

A digraph is \emph{twin-free} (resp. \emph{quasi-twin-free}) if it contains no twins (resp. no twins nor quasi-twins). 
A \emph{dominating set} $S$ of a digraph $D$ is a set of vertices such that any vertex not in $S$ is dominated by a vertex in $S$.
The smallest size of a dominating set of $D$ is called its domination number, and is denoted by $\gamma(D)$.
A \emph{locating set} $S$ of $D$ is a subset of vertices such that for every pair of vertices $x$ and $y$ not in $S$, there exists a vertex $s$ in $S$ that dominates exactly one vertex among $x$ and $y$.
The smallest size of a locating set of $D$ is denoted by $\gamma^L(D)$.
A \emph{locating-dominating set} of $D$ is a set of vertices that is both locating and dominating.
The minimum size of a locating-dominating set of $D$ is called the \emph{location-domination number}, denoted by $\gamma^{LD}(D)$.
Note that at most one vertex of $D$ is not dominated by a locating set of $D$, ensuring that: $$\gamma^L(D)\leq \gamma^{LD}(D)\leq \gamma^{L}(D)+1$$

Locating-dominating sets in digraphs have been introduced in \cite{LDdigraph1,LDdigraph2} and further studied in  \cite{Aline} where some upper bounds on the location-domination number of digraphs have been proved.
If $D$ contains at least one edge, then for any vertex $x$ that is not isolated, the set $V(D)\setminus\{x\}$ is locating-dominating and thus, if $D$ has order $n$, then $\gamma^{LD}(D)\leq n-1$. This bound is tight and a complete characterization of the digraphs reaching it is given in \cite{Aline}. 
However,  all these digraphs contain many pairs of twins. The authors of \cite{Aline} showed that any twin-free digraph $D$ of order $n$ satisfies $\gamma^{LD}(D)\leq \frac{4n}{5}+1$. Note that similar questions were first considered in the non-oriented case for which it is conjectured that any twin-free graph has a locating-dominating set of size $\frac{n} 2$ \cite{heia, Garijo}.

The authors of \cite{Aline} lowered the general upper bound $\frac{4n}{5}+1$ for special cases.
A digraph $D$ is a \emph{tournament} if there is exactly one arc between every pair of distinct vertices of $D$.
In particular, it is proven in \cite{Aline} that for a tournament $D$, the upper bound on the location-domination number can be lowered to $\gamma^{LD}(D)\leq \lceil \frac{n}{2}\rceil$.
We extend this last result to a larger class of digraphs.
A digraph $D$ is a \emph{local tournament} if it is simple and if the in-neighbourhood and the out-neighbourhood of every vertex of $D$ induce tournaments. The main result of the current paper is the following theorem.

\begin{theorem}
\label{th:main}
A connected local tournament $D$ of order $n$ satisfies $\gamma^{LD}(D) \leq \lceil \frac{n}{2} \rceil$.
\end{theorem}

Note that an edgeless graph of order $n$ is a local tournament for which $\gamma^{LD}(D) = n-1$. So if one removes the "connected" hypothesis in Theorem~\ref{th:main}, then the conclusion does not hold.

\

In \cite{Aline}, the authors asked if the general upper bound $\frac{4n}{5}+1$  for twin-free graphs
can be lowered to $\frac{2n}{3}$, for which some tight constructions are known (see Figure \ref{fig:tight} for the tight strongly connected example given in \cite{Aline}).
They proved that it is the case for quasi-twin-free acyclic digraphs. We extend these results by proving that if a quasi-twin-free digraph $D$ contains a {\em supervising vertex}, that is a vertex from which there exists a directed path to all the other vertices of the graph, then $\gamma^{LD}(D)\leq \frac{2n}{3}$. In particular, this bound is true (and asymptotically tight) for quasi-twin-free strongly connected digraphs and quasi-twin-free local-in semi-complete digraphs (i.e. digraphs, not necessarly simple, where two in-neighbours of a common vertex are connected).

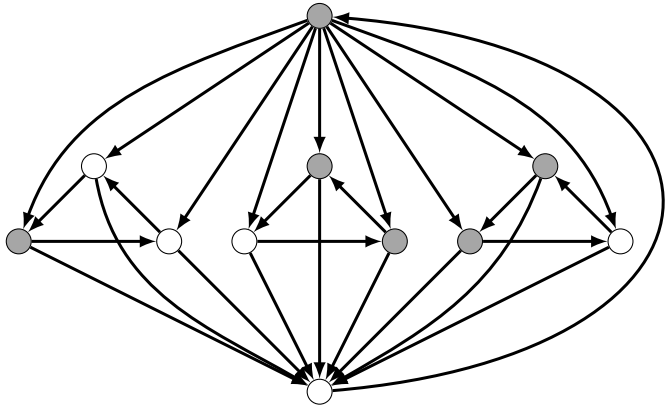
\begin{figure}[htpb!]
\begin{center}
\begin{tikzpicture}
\path (0,-1) node[draw,shape=circle] (s) {};
\path (0,4) node[draw,shape=circle,fill=gray!70] (t) {};

\path (1,1) node[draw,shape=circle,fill=gray!70] (b) {};
\path (-1,1) node[draw,shape=circle] (c) {};
\path (0,2) node[draw,shape=circle,fill=gray!70] (d) {};

\path (2,1) node[draw,shape=circle,fill=gray!70] (e) {};
\path (3,2) node[draw,shape=circle,fill=gray!70] (f) {};
\path (4,1) node[draw,shape=circle] (g) {};

\path (-4,1) node[draw,shape=circle,fill=gray!70] (h) {};
\path (-3,2) node[draw,shape=circle] (i) {};
\path (-2,1) node[draw,shape=circle] (j) {};

\draw[line width=0.4mm,>=latex,->] (s) .. controls +(6,0.5) and +(6,-0.5) .. (t);

\draw[line width=0.4mm,>=latex,->] (c) -- (b);
\draw[line width=0.4mm,>=latex,->] (b) -- (d);
\draw[line width=0.4mm,>=latex,->](d)--(c) ;

\draw[line width=0.4mm,>=latex,->](e) -- (g);
\draw[line width=0.4mm,>=latex,->](g) -- (f);
\draw[line width=0.4mm,>=latex,->](f) -- (e);

\draw[line width=0.4mm,>=latex,->](h) -- (j);
\draw[line width=0.4mm,>=latex,->](j) -- (i);
\draw[line width=0.4mm,>=latex,->](i) -- (h);

\draw[line width=0.4mm,>=latex,->](b) -- (s);
\draw[line width=0.4mm,>=latex,->](c) -- (s);
\draw[line width=0.4mm,>=latex,->](d) -- (s);
\draw[line width=0.4mm,>=latex,->](e) -- (s);
\path[line width=0.4mm,>=latex,->,draw](f) to[out=-110,in=30] (s);
\draw[line width=0.4mm,>=latex,->](g) -- (s);
\draw[line width=0.4mm,>=latex,->](h) -- (s);
\path[line width=0.4mm,>=latex,->,draw](i) to[out=-80,in=150] (s);
\draw[line width=0.4mm,>=latex,->](j) -- (s);

\draw[line width=0.4mm,>=latex,->](t) -- (b) ;
\draw[line width=0.4mm,>=latex,->](t) -- (c) ;
\draw[line width=0.4mm,>=latex,->](t) -- (d);
\draw[line width=0.4mm,>=latex,->](t) -- (e) ;
\draw[line width=0.4mm,>=latex,->](t) -- (f) ;
\path[line width=0.4mm,>=latex,->,draw](t) to[out=-20,in=110] (g) ;
\path[line width=0.4mm,>=latex,->,draw](t) to[out=-160,in=70] (h) ;
\draw[line width=0.4mm,>=latex,->](t) -- (i) ;
\draw[line width=0.4mm,>=latex,->](t) -- (j) ;

\end{tikzpicture}
\end{center}
\caption{A strongly connected twin-free and quasi-twin-free digraph of order $n$ with location-domination number $\frac{2(n-2)}{3}$, with a locating-dominating set in gray.}
\label{fig:tight}
\end{figure}

\
After giving some preliminary results on local tournaments in Section \ref{sec:preli}, we prove Theorem \ref{th:main} in two steps in Sections \ref{sec:round} and \ref{sec:notround}. In Section \ref{sec:supervising}, we prove the general upper bound for digraphs that contain a supervising vertex.

\section{Preliminaries}\label{sec:preli}
 
 In order to prove the upper bound on the location-domination number of local tournaments, we start by exposing some useful properties of local tournaments.

\begin{lemma}
A connected local tournament is twin-free.
\end{lemma}

\begin{proof}
Consider a connected local tournament $D$ and suppose by contradiction that there exist two distinct vertices $x,y$ of $D$ such that $x,y$ are twins. 

Suppose first that $x,y$ have some in-neighbours, then they have a common in-neighbour $z$. Since $x,y$ are both in the out-neighbourhood of $z$, there is an arc between them, a contradiction.

Suppose now that $x,y$ have no in-neighbours. 
Since $D$ is connected, we can consider a shortest (non necessarily directed) path $P=v_1,\ldots,v_k$ between $x,y$, such that $v_1=x$ and $v_k=y$. Note that $k\geq 3$.
By assumption, we have $v_1\to v_2$ and $v_k\to v_{k-1}$. Thus there exists $1<i<k$ such that $v_{i-1}\to v_i$ and $v_{i+1}\to v_i$. So $v_{i-1}$ and $v_{i+1}$ are both in the in-neighbourhood of $v_i$ and thus there is an arc between them, contradicting the fact that $P$ is a shortest path.
\end{proof}

The structure of local tournaments has been studied in~\cite{structure}.
In particular, the authors introduce the \emph{round decomposition}, which will be crucial in our proofs.
Therefore, let us have a closer look at the construction of such a decomposition.

We call a digraph $D$ on $r$ vertices \emph{round} if its vertices can be labelled $v_1,\ldots,v_{r}$ in such a way that $N^-(v_i)=\{v_{i-d^-(v_i)},\ldots,v_{i-1}\}$ and $N^+(v_i)=\{v_{i+1},\ldots,v_{i+d^+(v_i)}\}$ (indices are understood modulo $r$). One can easily check that a simple round digraph is a local tournament.

Let $R$ be a digraph on $r$ vertices and $L_1,\ldots,L_r$ be a collection of $r$ digraphs.
Then $R[L_1,\ldots,L_r]$ is the digraph obtained from $R$ by replacing each vertex $v_i$ of $R$ with $L_i$, and adding an arc from every vertex of $L_i$ to every vertex of $L_j$ if and only if $v_i\to v_j$ is an arc of $R$. 
A local tournament $D$ is said to be \emph{roundable} if it can be written as $R[T_1,\ldots,T_r]$ where $R$ is a round simple digraph on $r\geq 2$ vertices, and $T_1$,\ldots,$T_r$ are (non empty) strong tournaments. 
The decomposition $R[T_1,\ldots,T_r]$ is called a \emph{round decomposition} of $D$. 

Given a round decomposition $R[T_1,\ldots,T_r]$ of a roundable local tournament $D$,
if there is an arc between a vertex in $T_i$ and a vertex in $T_j$, then all the arcs between $T_i$ and $T_j$ are present in $D$, and we write $T_i \Rightarrow T_j$.
If $T_i \Rightarrow T_j$ and $i< j$, then for all $k\in\{i+1,\ldots,j-1\}$, we have $T_i\Rightarrow T_k$ and $T_k \Rightarrow T_j$.
If $T_i \Rightarrow T_j$ and $j< i$, then for all $k\in\{1,\ldots,j-1\}\cup\{i+1,\ldots,s\}$, we have $T_i\Rightarrow T_k$ and $T_k \Rightarrow T_j$.

Suppose that $D$ is a roundable connected local tournament that is not strongly connected. Then, as proved in~\cite{structure},  
there is a unique round decomposition $R[T_1,\ldots,T_r]$ of $D$ such that $T_i \Rightarrow T_j$ only happens if $i< j$. We often consider this particular round decomposition in this case, that we call the \emph{canonical}  round decomposition of $D$.
When $D$ is a roundable strong local tournament, there is a unique round decomposition up to cyclic permutation.

\section{Roundable local tournaments}
\label{sec:round}

This section is devoted to the study of roundable connected local tournaments.
Consider a roundable connected local tournament $D$ and a round decomposition $R[T_1,\ldots,T_r]$ of $D$  that is canonical if $D$ is not strong.

\begin{lemma}
\label{lem:next}
For $1\leq i< r$, we have  $T_i\Rightarrow T_{i+1}$. 
\end{lemma}

\begin{proof}
Suppose by contradiction that there exists $1\leq i< r$ such that we don't have  $T_i\Rightarrow T_{i+1}$. Then there is no arc $v_i \rightarrow v_{i+1}$ in $R$ and so by definition of round, we have $d^+(v_i)=0$. 
For $1\leq k\leq i< \ell \leq r$, there is no arc $v_k\rightarrow v_l$ as such an arc implies an arc $v_i\rightarrow v_\ell$. 
Since $v_i$ is a sink, the digraph $D$ is not strong and the decomposition is canonical.
So there is no edge between $V(T_1)\cup \cdots \cup V(T_i)$ and $V(T_{i+1})\cup \cdots \cup V(T_r)$, contradicting the fact that $D$ is connected.
\end{proof}

We define the sequence of integers $(i_k)_{0\leq k\leq t}$ as follows (see Figure~\ref{fig:decomposition}).
Let $i_0=0$.
If $i_k$ is defined and $i_k<r$, then we define $i_{k+1}$ as the greatest integer $i_{k}+1< i_{k+1}\leq r$ such that $T_{i_{k}+1} \Rightarrow T_{i_{k+1}}$ (note that this integer exists by Lemma~\ref{lem:next}).
This procedure stops when some $i_t$ is defined equal to $r$.
For $1\leq k\leq t$, 
let $D_k=D[T_{i_{k-1}+1}\cup\ldots\cup T_{i_{k}}]$.
Note that $D_k$ is a tournament for all $k$.

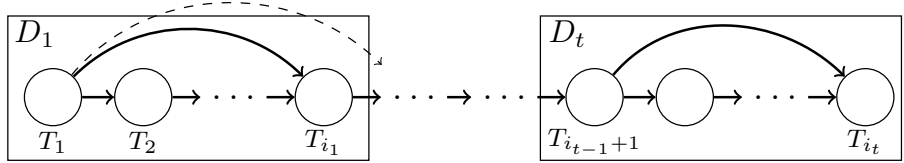
\begin{figure}[!h]
\begin{center}
\scalebox{1.2}{\begin{tikzpicture}
\foreach \x in {0,1,3,6,7,9}
{\node[circle,minimum size=18,draw](\x) at (\x,0) {};}
\foreach \x in {2,4,5,8}
{\node[minimum size=20](\x)  at (\x,0) {\large $\ldots$};}

\node at (0,-0.5) {\footnotesize $T_1$};
\node at (1,-0.5) {\footnotesize $T_2$};
\node at (3,-0.5) {\footnotesize $T_{i_1}$};
\node at (6,-0.5) {\footnotesize $T_{i_{t-1}+1}$};
\node at (9,-0.5) {\footnotesize $T_{i_t}$};
\path[->,draw,thick] (0) to (1);
\path[->,draw,thick] (0) to[out=45,in=135] (3);
\path[->,draw,dashed] (0) to[out=50,in=135] (4);
\path[->,draw,thick] (1) to (2);
\path[->,draw,thick] (2) to (3);
\path[->,draw,thick] (3) to (4);
\path[->,draw,thick] (4) to (5);
\path[->,draw,thick] (5) to (6);
\path[->,draw,thick] (6) to (7);
\path[->,draw,thick] (6) to[out=50,in=135] (9);
\path[->,draw,thick] (7) to (8);
\path[->,draw,thick] (8) to (9);

\draw (-0.5,-0.7) rectangle (3.5,0.9);
\node at (-0.2,0.7) {$D_1$};

\draw (5.4,-0.7) rectangle (9.4,0.9);
\node at (5.7,0.7) {$D_t$};
\end{tikzpicture}}
\end{center}
\caption{Illustration of the definition of the subgraphs $D_k$. The dashed arc means that there is no arc between the subgraphs, whereas thick arcs means that there are all the arcs between the two subgraphs.}
\label{fig:decomposition}
\end{figure}

Step by step, from $k=t$ to $k=1$, we define a set $S$ of vertices of $D$ that will be our candidate for a locating-dominating set of $D$. We consider the following three different cases:
\begin{itemize}
\item \emph{Case 1: $k=t$, or $k<t$ and $S\cap V(D_{k+1})$ is assumed to be a dominating set of $D_{k+1}$.}

If $|V(D_k)|$ is even, we add to $S$ a minimum locating-dominating set of  $D_k$, otherwise we add to $S$ a minimum locating set of $D_k$.

\item \emph{Case 2: $k<t$, the set $S\cap V(D_{k+1})$ is not assumed to be a dominating set of $D_{k+1}$, and $|V(T_{i_k})|=1$.}

Let $v$ be the unique vertex of $T_{i_k}$ and $D'_k=D_k\setminus\{v\}$. We add $v$ to $S$. Moreover we add to $S$ a minimum locating-dominating set of $D'_k$ when $|V(D'_k)|$ is even, or a minimum locating set of $D'_k$ when $|V(D'_k)|$ is odd.
(Note that the set of $D_k$ that is added to $S$, including $v$, is a dominating set of $D_k$ when $|V(D'_k)|$ is odd.)

\item \emph{Case 3: $k<t$, the set $S\cap V(D_{k+1})$ is not assumed to be a dominating set of $D_{k+1}$, and $|V(T_{i_k})|>1$.}

We add to $S$ a minimum locating-dominating set of $D_k$.
\end{itemize}

Note that in all the three cases we add to $S$ a set of vertices of $D_k$ that is a locating set of $D_k$, even in Case~2 since a locating set of $D'_k$ plus $v$ forms a locating set of $D_k$.

In the next lemmas we prove some properties of the set $S$.

\begin{lemma}
\label{claim:sizetwoset}
For $1\leq i \leq r$, the set $S\cap V(T_i)$ is a locating set of $T_i$, in particular if $T_i$ has at least 2 vertices, then $S\cap V(T_i)\neq \emptyset$.
\end{lemma}

\begin{proof}
Let $1\leq i \leq r$.
If $V(T_i)\setminus S$ has size $1$, then $S\cap V(T_i)$ is a locating set of $T_i$  (even if empty) and the conclusion holds. So we can assume that $V(T_i)\setminus S$ has size at least $2$.
Let $x$, $y$ be two distinct vertices of $V(T_i) \setminus S$.
Since $S$ is a locating set of $D_k$, there exists a vertex $s$ in $S\cap V(D_k)$ that dominates exactly one vertex among $x$ and $y$.
Since $R[T_1,\ldots,T_r]$ is a round decomposition, all the vertices of $D_k$ not in $T_i$ dominates either all or none of the vertices of $T_i$. So $s\in V(T_i)$. Therefore $S\cap V(T_i)$ is a locating set of $T_i$ and $S\cap V(T_i)\neq \emptyset$.
\end{proof}

\begin{lemma}
\label{claim:firstset}
For  $1\leq k\leq t$, if $S\cap V(D_k)$ is a dominating set of $D_k$, then $S\cap V(T_{i_{k-1}+1})\neq \emptyset$. Otherwise, there is exactly one vertex of $D_k$ that is not dominated by $S\cap V(D_k)$ and this vertex is in $V(T_{i_{k-1}+1})$.
\end{lemma}

\begin{proof}
Let  $1\leq k\leq t$. 
By definition of $i_k$,
all the vertices of $T_{i_{k-1}+1}$ dominates $D_k\setminus V(T_{i_{k-1}+1})$, 
so there is no arc from $V(D_k)\setminus V(T_{i_{k-1}+1})$ to $V(T_{i_{k-1}+1})$ (since $R$ is simple). 
Thus, if $S\cap V(D_k)$ is a dominating set of $D_k$, then $S$ must contains a vertex of $V(T_{i_{k-1}+1})$.

Suppose now that there exists a vertex of 
$D_k$ that is not dominated by $S\cap V(D_k)$.
Since $S\cap V(D_k)$ is a locating set of $D_k$, there is at most one vertex that is not dominated by $S\cap V(D_k)$, and thus exactly one such vertex $v$. 
If
there exists a vertex $u$ in
$S\cap V(T_{i_{k-1}+1})$, then $u$ dominates $D_k\setminus V(T_{i_{k-1}+1})$, so $v\in V(T_{i_{k-1}+1})$.
If there is no vertices in $S\cap V(T_{i_{k-1}+1})$, then the vertices of $T_{i_{k-1}+1}$ are not dominated by $S\cap V(D_k)$, so $V(T_{i_{k-1}+1})=\{v\}$.
\end{proof}

\begin{lemma}\label{claim:dominating}
There is at most one vertex of $D$ that is not dominated by $S$, and if it exists, it is a vertex of $T_1$.
\end{lemma}

\begin{proof} 
Assume by contradiction that there exists a vertex $x$ of $D$ that is not dominated  by $S$ and $x$ is not in $D_1$. Let $2\leq k \leq t$ such that $x\in D_k$. By Lemma~\ref{claim:firstset}, vertex $x$ is in $T_{i_{k-1}+1}$.
If $T_{i_{k-1}}$ contains a unique vertex $v$, then, at step $k-1$ of the construction of $S$ we are in Case 2 and $v$ is added to $S$.
If $T_{i_{k-1}}$ contains at least two vertices, then, by Lemma~\ref{claim:sizetwoset}, it contains a vertex of $S$.
In both cases, $S\cap V(T_{i_{k-1}})\neq \emptyset$. 
By Lemma~\ref{lem:next}, we have $T_{i_{k-1}} \Rightarrow T_{i_{k-1}+1}$, so $x$ is dominated by $S$, a contradiction.

Thus all the vertices that are not dominated by $S$ are in $D_1$. By Lemma~\ref{claim:firstset}, there is at most one such vertex and if it exists it is a vertex of $T_1$.
\end{proof}

\begin{lemma}\label{claim:locating}
Suppose that there exists a pair $\{x,y\}$ of vertices of $V(D)\setminus S$ not located by $S$. Then, $D$ is strongly connected, 
and one of $x,y$ is in $T_1$ and is the only vertex of $D_1$ not dominated by $S\cap V(D_1)$.
\end{lemma}

\begin{proof}
Let $1\leq k\leq t$, $ 1\leq \ell \leq t$ such that $x\in V(D_k)$ and $y\in V(D_\ell)$. 
Let $1\leq i\leq r$, $ 1\leq j \leq r$  such that $x\in V(T_i)$ and $y\in V(T_j)$.
Recall that for all $1\leq m\leq t$, $S\cap V(D_m)$ is a locating set of $D_m$, so
$k\neq \ell$ and $i\neq j$. Since a local tournament is simple, we can assume, w.l.o.g., that $y$ does not dominate $x$.

Suppose by contradiction that there exists $s\in S\cap V(D_{\ell})$ that dominates $y$.
Let $m$ be such that $s\in T_m$. 
Since $s,y$ are vertices of $D_\ell$, we have $1\leq m\leq j \leq r$.
Since $k\neq \ell$, we know that $T_m$, $T_j$ and $T_i$ appear in this order along the cyclic order $T_1, \ldots, T_r$ (with maybe $m=j$).
Since $y$ does not dominate $x$, the vertex $s$ does not dominate $x$ by definition of round.  
Thus $s$ separates $x$ and $y$, a contradiction. So $S\cap V(D_{\ell})$ does not dominate $y$ and by Lemma~\ref{claim:firstset},  
$y$ is the only vertex of $D_\ell$ that is not dominated by $S\cap V(D_\ell)$ and
$y\in T_{i_{\ell-1}+1}$. 

Suppose by contradiction that $\ell>1$. 
We use an argument that is similar to the proof of Lemma~\ref{claim:dominating}.
If $T_{i_{\ell-1}}$ contains a unique vertex $v$, then, at step $\ell-1$ of the construction of $S$ we are in Case 2 and $v$ is added to $S$.
If $T_{i_{\ell-1}}$ contains at least two vertices, then, by Lemma~\ref{claim:sizetwoset}, it contains a vertex of $S$.
In both cases, $S\cap V(T_{i_{\ell-1}})\neq \emptyset$. 
Let $s'\in S\cap V(T_{i_{\ell-1}})$.
By Lemma~\ref{lem:next}, we have $T_{i_{\ell-1}} \Rightarrow T_{i_{\ell-1}+1}$, so $s'$ dominates $y$.
Since $x$ and $y$ are not located by $S$, vertex $s'$ dominates $x$. 
Hence, $T_{i_{\ell-1}}$, $T_{i_{\ell-1}+1}$, $T_i$ appear in this order along the cyclic order $T_1, \ldots, T_r$ (with maybe $i_{\ell-1}=i$).
Since $x$ is dominated by $s'$ and not by $y$, we have $i_{\ell-1}=i$. So $k=\ell-1$. Thus both $x$ and $s'$ are in $V(T_{i_k})$ and $y$ is in $V(T_{i_k+1})$. 
At step $k$ of the construction of $S$, we are in Case~3, and so $S\cap V(D_k)$ is a dominating set of $D_k$.
By Lemma~\ref{claim:firstset}, the set $S$ contains a vertex of $V(T_{i_{k-1}+1})$.
By definition of $i_k$, this vertex dominates $x\in V(T_{i_k})$ but not $y\in V(T_{i_k+1})$, a contradiction. So $\ell=1$, $y\in T_{1}$ and $y$ is the only vertex of $D_1$ that is not dominated by $S\cap V(D_1)$.

Suppose by contradiction that $S$ is not strongly connected, then the round decomposition is canonical and there is no vertex of $D\setminus V(D_1)$ that can dominate $y$. So $y$ is not dominated by $S$ and by Lemma~\ref{claim:dominating}, vertex $x$ is dominated by $S$, a contradiction. So $S$ is strongly connected.
\end{proof}

We now need to slightly modify $S$ in a particular case.
If $S\cap V(D_1)$ is not a dominating set of $D_1$, then by Lemma~\ref{claim:firstset},
let $z \in T_1$ be the unique vertex of $D_1$ not dominated by $S\cap V(D_1)$. In this case, let $S^+=S\cup \{z\}$, otherwise, let $S^+=S$.
We are now able to prove the locating-dominating property of $S^+$.

\begin{lemma}\label{claim:LD}
The set $S^+$ is a locating-dominating set of $D$.
\end{lemma}

\begin{proof}
Suppose first that $S\cap V(D_1)$ is a dominating set of $D_1$. Then $S^+$ is a  dominating set of $D$ by Lemma~\ref{claim:dominating} and a locating set of $D$ by Lemma~\ref{claim:locating}.
Suppose now that $S\cap D_1$ is not a dominating set. Recall that $z$ is the only vertex of $D_1$ not dominated 
by $S\cap D_1$. By Lemma~\ref{claim:dominating}, there is a unique vertex that is not dominated by $S$ and this vertex is in $T_1$, so this vertex is $z$.
So $S^+$ is a  dominating set of $D$.
By Lemma~\ref{claim:locating}, vertex $z$ is in all the pairs of vertices of $V(D)\setminus S$ not located by $S$. Thus $S^+$ is locating.
\end{proof}

By Lemma~\ref{claim:LD}, we have defined a locating-dominating set of $D$. We now have to bound its size.
For that purpose we use the following theorem from~\cite{Aline}:

\begin{theorem}[\cite{Aline}]
\label{th:aline-tournament}
A tournament $D$ of order $n$ satisfies $\gamma^{LD}(D) \leq \lceil \frac{n}{2} \rceil$ and $\gamma^{L}(D) \leq \lfloor \frac{n}{2} \rfloor$.
\end{theorem}

We use Theorem~\ref{th:aline-tournament} to first bound the size of $S$ and then of $S^+$. 
At each step $k$ of the construction of $S$, since $D_k$ is a tournament, the theorem gives an upper bound on the size of the set that is added to $S$:

\begin{lemma}\label{lem:size}
For $1\leq k \leq t$, the set of vertices of $D_k$ that is added to $S$ at step $k$ of the construction of $S$ has size at most $\lceil |V(D_k)|/2 \rceil$ when it is assumed to be a locating-dominating set of $D_k$ and size at most $\lfloor |V(D_k)|/2 \rfloor$ when it is assumed to be a locating set of $D_k$.
\end{lemma}

\begin{proof}
Let $1\leq k \leq t$.
The lemma is clear by Theorem~\ref{th:aline-tournament} if we are in Cases~1 or~3 at step $k$ of the construction of $S$.
Consider now that we are in Case~2. If $|V(D'_k)|$ is even (i.e. $|V(D_k)|$ is odd), then $S\cap V(D_k)$ is composed of a minimum locating-dominating set of $D'_k$ plus the unique vertex of $T_{i_k}$. Thus, by Theorem~\ref{th:aline-tournament}, $S\cap V(D_k)$ is a locating-dominating set of $D_k$ of size at most $\lceil (|V(D_k)|-1)/2\rceil+1=\lceil |V(D_k)| /2 \rceil$. If $|V(D'_k)|$ is odd (i.e. $|V(D_k)|$ is even), then $S\cap V(D_k)$ is composed of a minimum locating set of $D'_k$ plus the unique vertex of $T_{i_k}$.
 Thus, similarly, $S\cap V(D_k)$ is a locating set of $D_k$ of size at most 
$\lfloor (|V(D_k)|-1)/2\rfloor+1=\lfloor |V(D_k)| /2 \rfloor$
\end{proof}

For every $1\leq k \leq t$, let $n_k=|V(D_k)|+ \cdots + |V(D_t)|$.
In the following lemma, we bound the size of a minimum locating-dominating set of a subgraph of $D$.

\begin{lemma}\label{claim:induction}
For $1\leq k \leq t$, the size of $S\cap (V(D_k)\cup \cdots \cup V(D_t))$ is at most $\lfloor \frac{n_k}{2} \rfloor $.
\end{lemma}

\begin{proof}
We prove the lemma by induction on $k$ from $t$ down to $1$.

At step $t$ of the construction of $S$, we are in Case~1. 
If $|V(D_t)|$ is even, then  $S\cap V(D_t)$ is a minimum locating-dominating set of $D_t$ of size at most $\lceil |V(D_t)|/2\rceil=\lfloor n_t/2\rfloor$ by Theorem~\ref{th:aline-tournament}. If $|V(D_t)|$ is odd, then $S\cap V(D_t)$ is a minimum locating set of $D_t$ of size at most 
$\lfloor n_t/2\rfloor$  by Theorem~\ref{th:aline-tournament}.

Let us now fix $1\leq k < t$. We assume that the lemma holds for all $j$ with $k+1\leq j \leq t$ and we prove that it holds for $k$. 
If $V(D_{k})$ has even size or $S\cap V(D_{k})$ is not assumed to dominate $D_{k}$, then by Lemma~\ref{lem:size}, $|S\cap V(D_{k})|\leq \lfloor |V(D_{k})|/2 \rfloor $ and it follows that the lemma holds for $k$. 
So we can now suppose that $V(D_{k})$ has odd size and $S\cap V(D_{k})$ is assumed to dominate $D_{k}$. 
Thus, we are in Case~2 or~3 at step $k$ and $S\cap V(D_{k+1})$ is not assumed to dominate $D_{k+1}$. Then, at step ${k+1}$, we are either in Case~1 with $V(D_{k+1})$ has odd size, or in Case~2 with $V(D_{k+1})$ has even size, $k+1<t$ and 
$S\cap V(D_{k+2})$ is not assumed to dominate $D_{k+2}$. In the second case, again it means that at step $k+2$ we are in the same situation, i.e. we are either in Case~1 with $V(D_{k+2})$ has odd size, or in Case~2 with 
$V(D_{k+2})$ has even size, $k+2<t$ and 
$S\cap V(D_{k+3})$ is not assumed to dominate $D_{k+3}$. We can repeat this argument until we actually encounter a step $\ell$ of the construction of $S$ such that
we are in Case~1 at step $\ell$  with $V(D_{\ell})$ has odd size (we assume $\ell$ is the first such step from $k+1$ towards $t$) .
Note that for $k+1 \leq j \leq \ell$, $S\cap V(D_{j})$ is not assumed to dominate $D_{j}$.
To summarize, we have: 
\begin{itemize}
\item $1\leq k < \ell \leq  t$
\item  $V(D_{k})$ and $V(D_{\ell})$ have odd sizes 
\item $V(D_{j})$ has even size for $k+1 \leq j \leq \ell-1$, 
\item $S\cap V(D_{k})$ is assumed to dominate $D_{k}$
\item $S\cap V(D_{j})$ is not assumed to dominate $D_{j}$ for $k+1 \leq j \leq \ell$,
\end{itemize} 

Thus, by Theorem~\ref{th:aline-tournament}, we obtain the following:
$$ |S\cap (V(D_{{k}})\cup\cdots\cup V(D_{\ell}))|$$$$\leq (|V(D_{k})|+1)/2 +|V(D_{k+1})\cup \cdots \cup V(D_{\ell-1})|/2 +(|V(D_{\ell})|-1)/2$$ $$\leq\lfloor |V(D_{k})\cup\ldots \cup V(D_{\ell})|/2\rfloor$$ 

If $\ell=t$, then the lemma holds for $k$ and we are done. If $\ell < t$, then we use the induction hypothesis that the lemma holds for $\ell+1$ to conclude.
\end{proof} 

A corollary of Lemma~\ref{claim:induction} is that the size of $S$ is at most $\lfloor n/2 \rfloor$. We are now able to bound $\vert S^+\vert$.

\begin{lemma}\label{claim:sizeS'}
The size of $S^+$ is at most $\lceil n/2 \rceil $.
\end{lemma}

\begin{proof}
By Lemma~\ref{claim:induction}, we have $|S| \leq \lfloor n/2 \rfloor$. Thus if $S= S^+$, then the lemma holds. Thus we can assume that 
$S\neq S^+$, i.e. $S\cap V(D_1)$ is not a dominating set of $D_1$ and $S^+=S\cup \{z\}$. If $n$ is odd, then 
 $|S^+|= |S|+1\leq \lfloor n/2 \rfloor +1 = \lceil n/2 \rceil $ and we are done. So we can assume that $n$ is even.

By Lemma~\ref{lem:size}, the size of $S\cap V(D_1)$ is at most $\lfloor |V(D_1)|/2\rfloor$ 

By construction of $S$, at step $1$, we are either in Case~1 with $|V(D_1)|$ is odd, or in Case~2 with $|V(D_1)|$ is even, $1<t$ and $S\cap V(D_2)$ is not assumed to dominate $D_2$.
We consider this two different cases below.

Consider first that we are in Case~1. Since $|V(D_1)|$ is odd and $n$ is even, we have $1<t$ and $n_2=n-|V(D_1)|$ is odd.
By Lemma~\ref{claim:induction},
the size of $S\cap (V(D_2)\cup \ldots \cup V(D_t))$ is at most $\lfloor n_2/2 \rfloor = (n_2-1)/2$.
Moreover $S\cap V(D_1)$ has size at most $\lfloor |V(D_1)|/2\rfloor=(|V(D_1)|-1)/2$. So 
in total $S$ has size at most $n/2-1$. Therefore $S^+$ has size at most $n/2$ and we are done.

Consider now that we are in Case~2 of the construction at step $1$. As in the proof of Claim \ref{claim:induction}, let $D_\ell$ be the first step from $2$ towards $t$
such that we are in Case~1 at step $\ell$ of the construction of $S$ with $V(D_{\ell})$ has odd size.
Then, we have:

\begin{itemize}
\item $2\leq \ell \leq  t$
\item $V(D_{j})$ has even size for $1 \leq j \leq \ell-1$, 
\item $V(D_{\ell})$ has odd size 
\item $S\cap V(D_{j})$ is not assumed to dominate $D_{j}$ for $1 \leq j \leq \ell$,
\end{itemize}

Thus, by Theorem~\ref{th:aline-tournament}, we obtain the following:
$$ |S\cap (V(D_{{1}})\cup\cdots\cup V(D_{\ell}))|\leq |V(D_{1})\cup \cdots \cup V(D_{\ell-1})|/2 +(|V(D_{\ell})|-1)/2$$

Since $n$ is even and $|V(D_{{1}})\cup\cdots\cup V(D_{\ell})|$ is odd, we have
$\ell<t$ and $n_{\ell+1}$ is odd. 
By Lemma~\ref{claim:induction}, 
the size of $S\cap (V(D_{\ell+1})\cup \cdots \cup V(D_t))$ is at most $\lfloor (n_{\ell+1})/2 \rfloor =(n_{\ell+1}-1)/2$.
So again $S$ has size at most $n/2-1$ and $S^+$ has size at most $n/2$.
\end{proof}

By combining the results of this section we obtain the following:

\begin{lemma}
\label{lem:roundboundS}
The set $S^+$ is a locating dominating set of $D$ of size at most $\lceil n/2 \rceil$.
Moreover, if $D$ is not strongly connected, then $S$ is a locating set of $D$ of size at most $\lfloor n/2 \rfloor$.
\end{lemma}

\begin{proof}
The first part is a consequence of Lemmas~\ref{claim:LD} and~\ref{claim:sizeS'}. 
The second part is a consquence of Lemmas~\ref{claim:locating} and~\ref{claim:induction}. 
\end{proof}

As a corollary, this proves Theorem \ref{th:main} for connected roundable local tournaments:

\begin{lemma}
\label{lem:roundbound}
A connected roundable local tournament $D$ of order $n$ satisfies $\gamma^{LD}(D)\leq \lceil \frac{n}{2} \rceil$. Moreover if $D$ is not strongly connected, then $\gamma^{L}(D)\leq \lfloor \frac{n}{2} \rfloor$. 
\end{lemma}

\section{Non-roundable local tournaments}\label{sec:notround}

In the previous section, we proved the upper bound $\gamma^{LD}(D) \leq \lceil n/2\rceil$ when $D$ is a connected roundable local tournament.
By Theorem~\ref{th:aline-tournament}, this result is also true for tournaments. In order to prove Theorem~\ref{th:main}, we can now restrict ourselves to connected local tournaments which are not tournaments and not roundable.

\bigskip
Consider a connected local tournament $D$ that is not a tournament and not roundable.
By  \cite[Corollary~3.2]{structure}, every connected local tournament that is not strong is roundable, so $D$ is strongly connected.
Then, by \cite[Corollary~3.2, Lemmas~3.4 and~3.5]{structure}, there exists a set of vertices $X$ of $V(D)$ such that:
\begin{itemize}
\item $D \setminus X$ is not strongly connected and $X$ is minimal for this property
\item $D \setminus X$ is a connected local tournament that is not a tournament
\item $D[X]$ is a tournament
\item Let $R[T_1,\ldots, T_r]$ be the canonical round decomposition of $D\setminus X$, then $r\geq 3$ and there are all the arcs from $V(T_r)$ to $X$, and from $X$ to $V(T_1)$.
\end{itemize}

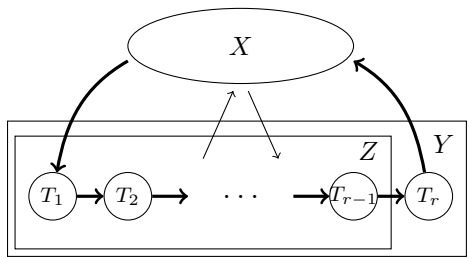
\begin{figure}[h]
\centering
\begin{tikzpicture}

\draw (2.5,2) ellipse (1.5cm and 0.5cm) node(x) {$X$};

\foreach \x in {0,1,4,5}
{\node[circle,minimum size=18,draw](\x) at (\x,0) {};}
\foreach \x in {2.5}
{\node[minimum size=20](\x)  at (\x,0) {\large $\ldots$};}

\node at (0,0) {\footnotesize $T_1$};
\node at (1,0) {\footnotesize $T_2$};
\node at (4,0) {\footnotesize $T_{r-1}$};
\node at (5,0) {\footnotesize $T_{r}$};

\path[->,draw,very thick] (0) to (1);
\path[->,draw,very thick] (1) to (1.8,0);
\path[->,draw,very thick] (3.2,0) to (4);
\path[->,draw,very thick] (4) to (5);
\path[->,draw,very thick] (1,1.8) to[out=-150,in=80] (0);
\path[->,draw,very thick] (5) to[out=100,in=-30] (4,1.8);
\path[->,draw] (2.6,1.4) to (3,0.5);
\path[->,draw] (2,0.5) to (2.4,1.4);

\draw (-0.5,-0.7) rectangle (4.5,0.8);
\node at (4.2,0.6) {$Z$};

\draw (-0.6,-0.8) rectangle (5.5,1);
\node at (5.2,0.7) {$Y$};

\end{tikzpicture}
\caption{Decomposition of non-roundable local tournaments}
\label{fig:nonroundable}
\end{figure}

Let $Y=V(D) \setminus X$ and $Z=Y\setminus V(T_r)$. See Figure \ref{fig:nonroundable} for an illustration of the decomposition of $D$.
Let $D_1=D[Y]$ and $D_2=D[Z]$. For $i\in\{1,2\}$, let $n_i$ be the order of $D_i$.
Note that since $r\geq 3$, the digraph $D_2$ is also a connected local tournament that is not strong with canonical round decomposition $R[T_1,\ldots, T_{r-1}]$.
So we can apply on $D_1$ and $D_2$ the method and results of Section~\ref{sec:round}.
For $i\in\{1,2\}$, let $S_i$ be the set that is defined on $D_i$ exactly as $S$ is defined on $D$ in Section~\ref{sec:round}.
By Lemma~\ref{lem:roundboundS}, the set  $S_i$ is a locating set of $D_i$ of size  $\lfloor n_i/2 \rfloor$. 
We will moreover need some particular properties obtained in Section~\ref{sec:round}, namely Lemmas~\ref{claim:sizetwoset} and~\ref{claim:dominating}.


Using these definitions of $S_1$ and $S_2$ we are now able to define a set of vertices $S$ of $D$ that will be our candidate for a locating-dominating set of $D$.
We consider the following four cases (see Figure \ref{fig:casesnonroundable} for an illustration):

\begin{itemize}
    \item \emph{Case 1: $\vert V(T_r)\vert = 1$ and $\vert X\vert = 1$}
    
    Let $S$ be the union of $V(T_r)$ and a minimum locating-dominating set of $D_2$
    
    \item \emph{Case 2: $\vert V(T_r)\vert = 1$ and $\vert X\vert > 1$}
    
    Let $S$ be the union of  $V(T_r)$, $S_2$, and a minimum locating set of $D[X]$.
    
    \item \emph{Case 3: $\vert V(T_r)\vert > 1$ and $\vert X\vert = 1$}
    
    Let $S$ be the union of $S_1$ and $X$.
    
        \item \emph{Case 4: $\vert V(T_r)\vert > 1$ and $\vert X\vert > 1$}
    
    Let $S$ be the union of $S_1$ and a minimum locating-dominating set of $D[X]$.
    \end{itemize}

\begin{figure}[h]
\centering
\scalebox{0.8}{\begin{tikzpicture}
\begin{scope}[shift={(-4,2.5)}]

\node[draw,shape=circle,minimum size=7,inner sep=0](x) at (2.5,2) {};
\node at (2.5,2.5) {$X$};

\foreach \x in {0,1,4}
{\node[circle,minimum size=18,draw](\x) at (\x,0) {};}
\foreach \x in {2.5}
{\node[minimum size=20](\x)  at (\x,0) {\large $\ldots$};}

\node at (0,0) {\footnotesize $T_1$};
\node at (1,0) {\footnotesize $T_2$};
\node at (4,0) {\footnotesize $T_{r-1}$};
\node[draw,shape=circle,fill=gray!70,minimum size=7,inner sep=0](5) at (5,0) {};
\node at (5,-0.5) {\footnotesize $T_{r}$};

\path[->,draw,very thick] (0) to (1);
\path[->,draw,very thick] (1) to (1.8,0);
\path[->,draw,very thick] (3.2,0) to (4);
\path[->,draw,very thick] (4) to (5);
\path[->,draw,very thick] (x) to[out=-170,in=60] (0);
\path[->,draw,thick] (5) to[out=110,in=-20] (x);
\path[->,draw] (x) to (2,0.5);
\path[->,draw] (3,0.5) to (x);

\draw (-0.5,-0.8) rectangle (4.5,0.8);
\node at (4.2,0.6) {$Z$};
\node at (2,-0.5) {LD-set};

\draw (-0.6,-0.9) rectangle (5.5,1);
\node at (5.2,0.7) {$Y$};

\node at (2.5,-1.5) {Case 1:  $|V(T_r)|=1$ and $|X|=1$};
\end{scope}

\begin{scope}[shift={(4,2.5)}]

\draw (2.5,2) ellipse (1.5cm and 0.5cm) node(x) {};
\node at (2.5,2.2) {$X$};
\node at (2.5,1.8) {L-set};

\foreach \x in {0,1,4}
{\node[circle,minimum size=18,draw](\x) at (\x,0) {};}
\foreach \x in {2.5}
{\node[minimum size=20](\x)  at (\x,0) {\large $\ldots$};}

\node at (0,0) {\footnotesize $T_1$};
\node at (1,0) {\footnotesize $T_2$};
\node at (4,0) {\footnotesize $T_{r-1}$};
\node[draw,shape=circle,fill=gray!70,minimum size=7,inner sep=0](5) at (5,0) {};
\node at (5,-0.5) {\footnotesize $T_{r}$};

\path[->,draw,very thick] (0) to (1);
\path[->,draw,very thick] (1) to (1.8,0);
\path[->,draw,very thick] (3.2,0) to (4);
\path[->,draw,very thick] (4) to (5);
\path[->,draw,very thick] (1,1.8) to[out=-150,in=80] (0);
\path[->,draw,very thick] (5) to[out=100,in=-30] (4,1.8);
\path[->,draw] (2.6,1.4) to (3,0.5);
\path[->,draw] (2,0.5) to (2.4,1.4);

\draw (-0.5,-0.8) rectangle (4.5,0.8);
\node at (4.2,0.6) {$Z$};
\node at (2,-0.5) {L-set $S2$};

\draw (-0.6,-0.9) rectangle (5.5,1);
\node at (5.2,0.7) {$Y$};

\node at (2.5,-1.5) {Case 2:  $|V(T_r)|=1$ and $|X|>1$};
\end{scope}

\begin{scope}[shift={(-4,-2.5)}]
\node[draw,shape=circle,minimum size=7,inner sep=0,fill=gray!70](x) at (2.5,2) {};
\node at (2.5,2.5) {$X$};

\foreach \x in {0,1,4,5}
{\node[circle,minimum size=18,draw](\x) at (\x,0) {};}
\foreach \x in {2.5}
{\node[minimum size=20](\x)  at (\x,0) {\large $\ldots$};}

\node at (0,0) {\footnotesize $T_1$};
\node at (1,0) {\footnotesize $T_2$};
\node at (4,0) {\footnotesize $T_{r-1}$};
\node at (5,0) {\footnotesize $T_{r}$};

\path[->,draw,very thick] (0) to (1);
\path[->,draw,very thick] (1) to (1.8,0);
\path[->,draw,very thick] (3.2,0) to (4);
\path[->,draw,very thick] (4) to (5);
\path[->,draw,very thick] (x) to[out=-170,in=60] (0);
\path[->,draw,very thick] (5) to[out=100,in=-10] (x);

\path[->,draw] (x) to (2,0.5);
\path[->,draw] (3,0.5) to (x);

\draw (-0.5,-0.5) rectangle (4.5,0.8);
\node at (4.2,0.6) {$Z$};

\draw (-0.6,-1) rectangle (5.5,1);
\node at (5.2,0.7) {$Y$};
\node at (2.5,-0.7) {L-set $S1$};
\node at (2.5,-1.5) {Case 3:  $|V(T_r)|>1$ and $|X|=1$};
\end{scope}

\begin{scope}[shift={(4,-2.5)}]

\draw (2.5,2) ellipse (1.5cm and 0.5cm) node(x) {};
\node at (2.5,2.2) {$X$};
\node at (2.5,1.8) {LD-set};

\foreach \x in {0,1,4,5}
{\node[circle,minimum size=18,draw](\x) at (\x,0) {};}
\foreach \x in {2.5}
{\node[minimum size=20](\x)  at (\x,0) {\large $\ldots$};}

\node at (0,0) {\footnotesize $T_1$};
\node at (1,0) {\footnotesize $T_2$};
\node at (4,0) {\footnotesize $T_{r-1}$};
\node at (5,0) {\footnotesize $T_{r}$};

\path[->,draw,very thick] (0) to (1);
\path[->,draw,very thick] (1) to (1.8,0);
\path[->,draw,very thick] (3.2,0) to (4);
\path[->,draw,very thick] (4) to (5);
\path[->,draw,very thick] (1,1.8) to[out=-150,in=80] (0);
\path[->,draw,very thick] (5) to[out=100,in=-30] (4,1.8);
\path[->,draw] (2.6,1.4) to (3,0.5);
\path[->,draw] (2,0.5) to (2.4,1.4);

\draw (-0.5,-0.5) rectangle (4.5,0.8);
\node at (4.2,0.6) {$Z$};

\draw (-0.6,-1) rectangle (5.5,1);
\node at (5.2,0.7) {$Y$};
\node at (2.5,-0.7) {L-set $S1$};
\node at (2.5,-1.5) {Case 4:  $|V(T_r)|>1$ and $|X|>1$};
\end{scope}

\end{tikzpicture}}
\caption{Construction of a locating-dominating sets in a non-roundable local tournament. Gray vertices are in the locating-dominating sets.}
\label{fig:casesnonroundable}
\end{figure}
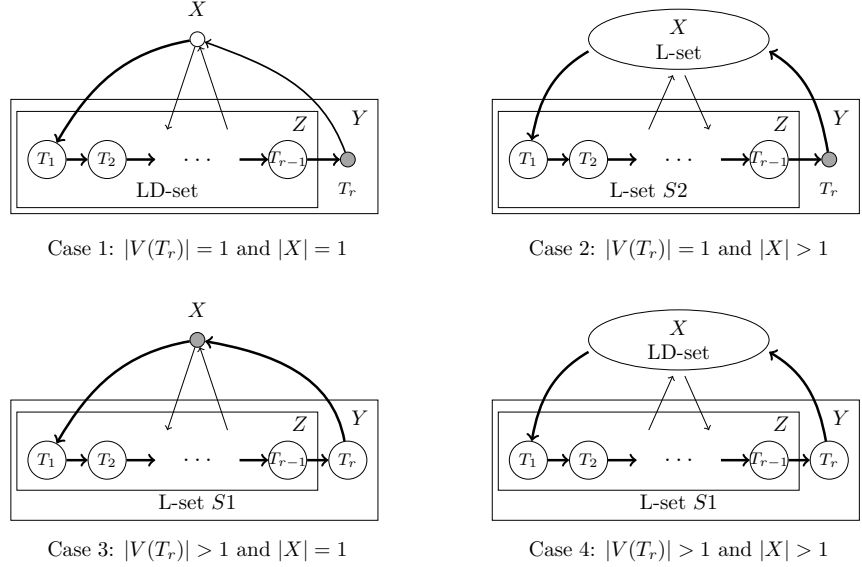


\begin{lemma}
\label{lem:ldnotround}
The set $S$ is a locating-dominating set of $D$.
\end{lemma}

\begin{proof}
    \noindent \emph{(Case 1)}
Let $x$ be the only vertex of $X$ and $t$ be the only vertex of $T_r$.
The set $S$ contains a locating-dominating set of $D_2$ so all the vertices of $Y\setminus S$ are dominated by $S$. Moreover $x$ is dominated by $t$. So the set $S$ is a dominating set of $D$.

Every pair of vertices of $Y\setminus S$ is separated by $S$. Since the round decomposition is canonical, there is no arc from $t$ to $Z$, so $t$ separates $x$ from any vertex of $Y\setminus S$. Hence $S$ is a locating set of $D$.

    \medskip
   
    \noindent \emph{(Case 2)}
As previously, let $t$ be the only vertex of $T_r$.
By Lemma~\ref{claim:dominating} applied on $S_2$, there is at most one vertex of $Z$ that is not dominated by $S_2$, and if it exists, it is a vertex of $T_1$.
However, by assumption, $\vert X \vert > 1$, so a locating set of $D[X]$ contains at least one vertex of $X$. Thus $S$ contains at least one vertex $s$ of $X$.
Since there are all the arcs from $X$ to $V(T_1)$, all the vertices of $T_1$ are dominated by $s$.
All the vertices of $X$ are dominated by $t$. So $S$ is a dominating set of $D$.

Every pair of vertices of $X\setminus S$ and every pair of vertices of $Z\setminus S$ are separated by $S$.
There are all the arcs from $t$ to $X$ and no arc from $t$ to $Y\setminus \{t\}$. 
Hence, $t$ separates the vertices of $X$ from the vertices of $Z$, and $S$ is a locating set of $D$.

    \medskip
   
    \noindent \emph{(Case 3)}
Let $x$ be the only vertex of $X$.
By Lemma~\ref{claim:dominating} applied on $S_1$, there is at most one vertex of $Y$ that is not dominated by $S_1$, and if it exists, it is a vertex of $T_1$.
Since $x\in S$ and there are all the arcs between $x$ and $T_1$, then every vertex of $T_1$ is dominated by $x$. Hence $S$ is a dominating set of $D$.

By definition of $S$, all the pairs of vertices of $Y\setminus S$ are separated by $S$ and $x\in S$.
So $S$ is a locating set of $D$.
    
    \medskip
   
    \noindent \emph{(Case 4)}
Note that, since $S$ contains a locating set of $X$ and  $\vert X \vert > 1$, the set $S$ contains at least one vertex of $X$. Then the proof that $S$ is a dominating set of $D$ is exactly the same as in Case~3.

By definition of $S$, every pair of vertices of $X\setminus S$ and every pair of vertices of $Y\setminus S$ are separated by $S$.
Now consider $x\in X\setminus S$ and $y \in Y\setminus S$ and let us show that $x$ and $y$ are separated by $S$. We consider the following two cases:
\begin{itemize}
    \item \emph{Case A : $y\in T_r$}
    
    Since $S$ contains a dominating set of $D[X]$, there is a vertex $x'$ of $X\cap S$ that dominates $x$ in $D$.
    Since there are all the arcs from $V(T_s)$ to $X$, there is an arc from $y$ to $x'$. Since we are considering a simple digraph, there is no arc from $x'$ to $y$, so $x'$ separates $x$ and $y$.
    
    \item \emph{Case B : $y\notin T_r$}
    
    By Lemma~\ref{claim:sizetwoset} applied on $S_1$, there exists a vertex $y'$ in $S\cap T_r$. 
        Since the round decomposition is canonical, there is no arc from $y'$ to $y$.
                Plus, there are all the arcs from $V(T_r)$ to $X$, and there is an arc from $y'$ to $x$.
        Hence $y'$ separates $x$ and $y$.
        \end{itemize}
Therefore $S$ is a locating-dominating set of $D$.
\end{proof}

\begin{lemma}
\label{lem:sizenotround}
The size of $S$ is at most $\lceil \frac{n}{2} \rceil$.
\end{lemma}

\begin{proof}

    \noindent \emph{(Case 1)}
        By Lemma~\ref{lem:roundbound}, a minimum  locating-dominating set of $D_2$ has size at most
     $\lceil \frac{n_2}{2}\rceil$. So $S$ has size at most $\lceil \frac{n_2}{2}\rceil+1 = \lceil \frac{n-2}{2} \rceil + 1 = \lceil \frac{n}{2} \rceil$.
    
    \medskip
    
    \noindent \emph{(Case 2)}
   By Lemma~\ref{lem:roundboundS}, the size of $S_2$ is at most $\lfloor \frac{n_2}{2} \rfloor = \lfloor \frac{n-1-\vert X\vert}{2} \rfloor$.
   By Theorem~\ref{th:aline-tournament}, a minimum locating set of $D[X]$ has size at most $\lfloor \frac{\vert X\vert}{2}\rfloor$.
    So $S$ has size at most  $ \lfloor \frac{n-1-\vert X\vert}{2} \rfloor + \lfloor \frac{\vert X\vert}{2}\rfloor +1 \leq \lceil \frac{n}{2} \rceil$ (check all the parity cases for the last inequality).

    \medskip
       
    \noindent \emph{(Case 3)}
    By Lemma~\ref{lem:roundboundS}, the size of $S_1$ is at most $\lfloor \frac{n_1}{2} \rfloor$.
So $S$ has size at most $\lfloor \frac{n-1}{2} \rfloor + 1 \leq \lceil \frac{n}{2} \rceil$.
    
    \medskip
       
    \noindent \emph{(Case 4)}
    By Lemma~\ref{lem:roundboundS}, the size of $S_1$ is at most $\lfloor \frac{n_1}{2} \rfloor$.
    By Theorem~\ref{th:aline-tournament},  a minimum locating-dominating set of $D[X]$ has size at most $\lceil \frac{\vert X\vert}{2}\rceil$.
    So $S$ has size at most $\lfloor \frac{n - \vert X\vert}{2}\rfloor + \lceil \frac{\vert X\vert}{2} \rceil \leq \lceil \frac{n}{2}\rceil$ (check all the parity cases for the last inequality).
\end{proof}

By combining the results of this section we obtain the following:

\begin{lemma}
\label{lem:notround}
A  connected local tournament  $D$ of order $n$ that is not roundable satisfies $\gamma^{LD}(D)\leq \lceil \frac{n}{2} \rceil$.
\end{lemma}

\begin{proof}
By Lemmas~\ref{lem:ldnotround} and~\ref{lem:sizenotround}, the set $S$ is a locating dominating set of $D$ of size at most $\lceil n/2 \rceil$.
\end{proof}

Theorem~\ref{th:main} is a direct consequence of Lemmas~\ref{lem:roundbound} and~\ref{lem:notround}.

\section{Supervising vertex}\label{sec:supervising}

From now on, we consider digraphs that are not necessarily simple.
In a digraph $D$, a \emph{supervising} vertex is a vertex $s$  of $D$ such that, for any vertex $v$, there exists a directed path from $s$ to $v$. In this section we prove the following theorem.

\begin{theorem}
\label{th:supervising}
Let $D$ be a twin-free digraph on $n$ vertices containing a supervising vertex, then 
$\gamma^{LD}(D) \leq \frac{3n}{4}$.
Moreover, if $D$ is quasi-twin-free, then 
$\gamma^{LD}(D) \leq \frac{2n}{3}$.
\end{theorem}


To prove Theorem \ref{th:supervising}, we will adapt the method used in \cite{heia} and \cite{Aline} to prove general upper bounds on $\gamma^{LD}$.
Let $S$ be a set of vertices of a digraph $D$. The {\em $S$-partition} of $D$, denoted $\mathcal P_S$, is the partition of $V(D)\setminus S$ where two vertices are in the same part if and only if they have the same set of in-neighbours in $S$.
We have the following lemma.

\begin{lemma}
\label{th:aline}
Let $D$ be a twin-free digraph on $n$ vertices and $S$ a dominating set of $D$ such
that $|\mathcal P_S|\geq |S|-1$. Then, 
$\gamma^{LD}(D) \leq \frac{3n}{4}$.
Moreover, if $D$ is quasi-twin-free, then 
$\gamma^{LD}(D) \leq \frac{2n}{3}$.
\end{lemma}

This result is proved in \cite{Aline} when $|\mathcal P_S|\geq |S|$ (Theorem~8 for $x=1$). 
The proof can be adapted if we only have $|\mathcal P_S|\geq |S|-1$.

\begin{proof}
Let $\mathcal P_S=P_1 \cup\cdots\cup P_{n_1}\cup Q_1\cup\cdots \cup Q_{n_2}$, where $P_1,\ldots, P_{n_1}$ are the parts of size~$1$ and $Q_1,\ldots, Q_{n_2}$ are the parts of size at least~$2$.

If $n_2=0$, then $S$ is a locating-dominating set of $D$. Since all the parts of $\mathcal P_S$ have size 1, $|V(D)|=|S|+|\mathcal P_S|\geq 2|S|-1$. Thus $|S|\leq \frac{n+1}{2}$ and we are done. Thus in the following we assume that $n_2>0$.

We assume that $S$ is maximal with the property that $\mathcal P_S$ has at least $|S|-1$ parts (this is ensured by adding vertices to $S$ while this property holds).

Now, let $X_1 = S\cup P_1 \cup\cdots\cup P_{n_1}$. We have the following property:

\begin{claim}\label{claim:D1}
Two vertices in $V(D)\setminus X_1$ are located by $X_1$, unless they form a pair of quasi-twins.
\end{claim}
\claimproof If two vertices are in different parts of $\mathcal P_S$, they are located by some vertices in $S$. Thus, by contradiction, let $q_1$ and $q_2$ be two vertices of $V(D) \setminus X_1$ belonging to some part $Q_i$ of $\mathcal P_S$ that are not quasi-twins but are not located by $X_1$. Since $D$ is twin-free, there is a vertex $q_3$ in $V(D)\setminus S$ that can locate $q_1$ and $q_2$: without loss of generality $q_3$ is an in-neighbour of $q_1$ but not $q_2$. By our assumption $q_3 \notin X_1$. Now, consider $S'= S\cup \lbrace q_3\rbrace$ and the corresponding $S'$-partition $\mathcal P_{S'}$ of $V(D)\setminus S'$. Since $q_3\in\bigcup_i Q_i$, any part of $\mathcal P_S$ still correspond to some part in $\mathcal P_{S'}$. But $Q_i$ has been split into two parts so $P_{S'}$ has at least one more part than $\mathcal P_S$, and thus, $|\mathcal P_{S'}|\geq |S'|-1$. This contradicts the choice of $S$, which we assumed to be maximal with this property.~\smallqed

\medskip

Since $X_1$ is a dominating set, Claim~\ref{claim:D1} shows that in the absence of quasi-twins, $X_1$ is locating-dominating. Next claim is proved in \cite[Claim 8.B]{Aline}  to deal with quasi-twins. 

\begin{claim}[\cite{Aline}]\label{claim:quasi-twins}
Any two pairs of quasi-twins in $V(D)\setminus X_1$ are disjoint.
\end{claim}

For each pair of quasi-twins in $V(D)\setminus X_1$, we add one of the vertices of the pair in $X_1$. By Claims \ref{claim:D1} and \ref{claim:quasi-twins}, the resulting set $X_1'$ is a locating-dominating set and has size at most $|S|+n_1+(n-|S|-n_1)/2=(n+|S|+n_1)/2$.

Consider now the set $X_2$ of size $n-n_1-n_2$ consisting of $V(D)$ without one vertex from each part of $\mathcal P_S$. Then all the vertices of $V(D)\setminus X_2$ are located and dominated by $S$ and thus $X_2$ is a locating dominating set.

Assume now that $D$ has no quasi-twins. Then $X_1$ and $X_2$ are two locating-dominating sets of $D$. If $|X_2|\leq 2n/3$ we are done. Thus we assume that $|X_2|>2n/3$ which means that $|\mathcal P_S|<n/3$.
Therefore,

\begin{align*}
|X_1| &=|S|+n_1 \\
&\leq |\mathcal P_S|+n_1 +1\\
&\leq |\mathcal P_S| + (n_1+ n_2) \quad \quad \text{since } n_2\geq 1\\
&\leq 2|\mathcal P_S| \\
&\leq 2n/3
\end{align*}
and we are done.

If $D$ has some quasi-twins, we use the locating-dominating sets $X_1'$ and $X_2$. Again, if $|X_2|\leq 3n/4$, we are done. So, assume that $|X_2|> 3n/4$. Then, $|\mathcal P_S|<n/4$. 

Therefore,
\begin{align*}
|X_1'| &=\frac{|S|+n+n_1}{2}\\
&\leq \frac{|\mathcal P_S|+1+n+n_1}{2}\\
&\leq \frac{|\mathcal P_S|+n+n_1+n_2}{2}\\
&\leq |\mathcal P_S|+\frac{n}{2}\\
&\leq \frac{n}{4}+\frac{n}{2}\\
&= \frac{3n}{4}\;
\end{align*}
and we are done.
\end{proof}

To apply Lemma~\ref{th:aline}, we prove that such a set $S$ exists when there is a supervising vertex.
\begin{lemma}
\label{lem:supervising}
If a digraph $D$ contains a supervising vertex $s$, then, there exists a dominating set $S$ such that $|\mathcal P_{S}|\geq |S|-1$.
\end{lemma}

\begin{proof}
Consider a supervising vertex $s$ of $D$. For $i\geq 0$, let $V_i$ be the set of vertices of $D$ such that the shortest directed path from $s$ to $v$ has length $i$. Let $k$ be the smallest integer such that $V_{k+1}$ is empty. Since $s$ is supervising, $V_0,\ldots,V_k$ form a partition of $V(D)$ where the vertices are sorted according to their distance from $s$.
If $k=0$, then $D$ contains $s$ as a unique vertex and the set $S=\emptyset$ satisfies the lemma. So from now we assume that $k>0$.

We build the set $S$ by the following method. 
Let $S$ be a set of vertices of $V_{k-1}$ that dominates $V_k$ and that is minimal for this property.
Then, step by step, for $i=1$ to $k-1$, we assume that $S$ is a set of vertices of $V_{k-1}\cup \ldots \cup V_{k-i}$ that dominates $V_{k}\cup \ldots\cup V_{k-i+1}$ and we add to $S$ a set of vertices of $V_{k-i-1}$ that dominates $V_{k-i}\setminus S$ and that is minimal for this property. We continue the process until $i=k-1$. 

At each step $i$, by minimality of the choosen set, when a vertex $v$ of $V_{k-i-1}$ is added to $S$ one can choose a vertex $f(v)$ in $V_{k-i}\setminus S$ whose in-neighbourhood in $V_{k-i-1}\cap S$ is exactly $v$. By doing so, at the end of the procedure, each vertex of $f(S)$ has different in-neighbours in $S$. So  $|\mathcal P_{S}|\geq |S|$. 

Finally if $s$ is not already dominated by $S$, then we add it to $S$ so that $S$ is a dominating set of $V(D)$. This might increase the cardinality of $S$ by $1$ and in the end $S$ is a dominating set such that $|\mathcal P_S|\geq |S|-1$.
\end{proof}

Then Theorem \ref{th:supervising} is direct consequence of Lemmas~\ref{th:aline} and~\ref{lem:supervising}.

Since all the vertices in a strongly connected digraph are supervising, we have the following corollary:

\begin{corollary}
Let $D$ be a twin-free strongly connected digraph on $n$ vertices, then 
$\gamma^{LD}(D) \leq \frac{3n}{4}$.
Moreover, if $D$ is quasi-twin-free, then 
$\gamma^{LD}(D) \leq \frac{2n}{3}$.
\end{corollary}

Note that the second bound is asymptotically tight (see Figure \ref{fig:tight}). 
A digraph is called \emph{semicomplete} if there is at least one arc between every pair of vertices. A digraph $D$ is \emph{locally in-semicomplete} if the in-neighbourhood of every vertex $x$ of $D$ induces a semicomplete digraph. Note that semicomplete digraphs are a generalization of tournaments and thus, locally in-semicomplete digraphs naturally generalize local tournaments.

\begin{lemma}
\label{lem:localIN}
A locally in-semicomplete digraph has a supervising vertex.
\end{lemma}

\begin{proof}
Suppose by contradiction, that there exists a locally in-semicomplete digraph $D$ with no supervising vertex.
Consider a vertex of $D$, such that the set $S$ of vertices $v$ of $V(D)$ for which there exists a directed path from $s$ to $v$ has maximum size, i.e. $s$ is \emph{supervising} a set of vertices of maximum size.
Note that $s\in S$ since $s$ forms a directed path of length zero from itself to itself.
Since $D$ contains no supervising vertices, we know that $V(D)\setminus S$ is nonempty.
Plus, all the edges between $V(D)\setminus S$ and $S$ are oriented from $V(D)\setminus S$ to $S$.
Since $D$ is connected, there exists at least such an edge from a vertex $u$ of $V(D)\setminus S$ to a vertex of $S$.
Recall that $u$ has no in-neighbour in $S$.
Let $v$ be a vertex of $S$ that is an out-neighbour of $u$ and such that the length of a shortest directed path $P$ from $s$ to $v$ has minimum length. 
By the choice of $v$, no vertex of $P$ distinct from $v$ is an out-neighbour of $u$.
Note that $v$ is distinct from $s$, since otherwise $u$ is supervising more vertices than $s$, contradicting the choice of $s$.
Let $w$ be the in-neighbour of $v$ along $P$. Both $w$ and $u$ are in the in-neighbourhood of $v$. Hence, by in-semicompleteness of $D$, there must be an arc between them, which is a contradiction.
\end{proof}

A consequence of Lemma~\ref{lem:localIN} and Theorem~\ref{th:supervising} is the following corollary.

\begin{corollary}
\label{cor:localIN}
Let $D$ be a twin-free locally in-semicomplete digraph  on $n$ vertices, then 
$\gamma^{LD}(D) \leq \frac{3n}{4}$.
Moreover, if $D$ is quasi-twin-free, then 
$\gamma^{LD}(D) \leq \frac{2n}{3}$.
\end{corollary}

Consider the digraph $D$ obtained from a vertex whose out-neighbourhood is made of $k$ disjoint oriented triangles (see Figure~\ref{fig:localIN}). On this example, $D$ is a quasi-twin-free locally in-semicomplete digraph on $3k+1$ vertices. Note that $\gamma^{LD}(D)=2k$ and $\frac{2n}{3}=2k+1$. So we are at distance $1$ from the bound given by Corollary~\ref{cor:localIN} and thus asymptotically tight.

\begin{figure}[!ht]
\center
\begin{tikzpicture}
\path (0,4) node[draw,shape=circle,fill=gray!70] (t) {};

\path (1,1) node[draw,shape=circle,fill=gray!70] (b) {};
\path (-1,1) node[draw,shape=circle] (c) {};
\path (0,2) node[draw,shape=circle,fill=gray!70] (d) {};

\path (2,1) node[draw,shape=circle,fill=gray!70] (e) {};
\path (3,2) node[draw,shape=circle,fill=gray!70] (f) {};
\path (4,1) node[draw,shape=circle] (g) {};

\path (-4,1) node[draw,shape=circle,fill=gray!70] (h) {};
\path (-3,2) node[draw,shape=circle] (i) {};
\path (-2,1) node[draw,shape=circle] (j) {};

\draw[line width=0.4mm,>=latex,->] (c) -- (b);
\draw[line width=0.4mm,>=latex,->] (b) -- (d);
\draw[line width=0.4mm,>=latex,->](d)--(c) ;

\draw[line width=0.4mm,>=latex,->](e) -- (g);
\draw[line width=0.4mm,>=latex,->](g) -- (f);
\draw[line width=0.4mm,>=latex,->](f) -- (e);

\draw[line width=0.4mm,>=latex,->](h) -- (j);
\draw[line width=0.4mm,>=latex,->](j) -- (i);
\draw[line width=0.4mm,>=latex,->](i) -- (h);

\draw[line width=0.4mm,>=latex,->](t) -- (b) ;
\draw[line width=0.4mm,>=latex,->](t) -- (c) ;
\draw[line width=0.4mm,>=latex,->](t) -- (d);
\draw[line width=0.4mm,>=latex,->](t) -- (e) ;
\draw[line width=0.4mm,>=latex,->](t) -- (f) ;
\path[line width=0.4mm,>=latex,->,draw](t) to[out=-20,in=110] (g) ;
\path[line width=0.4mm,>=latex,->,draw](t) to[out=-160,in=70] (h) ;
\draw[line width=0.4mm,>=latex,->](t) -- (i) ;
\draw[line width=0.4mm,>=latex,->](t) -- (j) ;

\end{tikzpicture}
\caption{Disjoint oriented triangles forming the out-neighbourhood of an extra vertex. A locating dominating set of minimum size is given by gray vertices.}
\label{fig:localIN}
\end{figure}
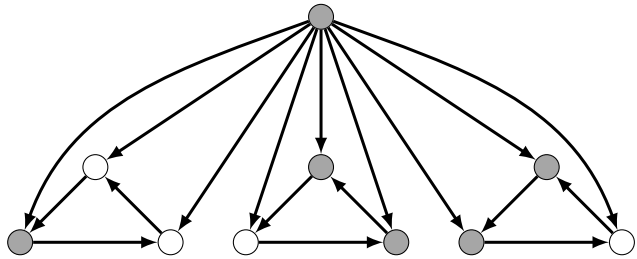

Note that a similar method cannot be applied for locally out-semicomplete digraphs (digraphs where all the out-neighbourhoods are semicomplete). Indeed, there are twin-free locally out-semicomplete digraphs for which the minimum dominating set has size $2(n-1)/3$ 
(see for example the reverse of the digraph of Figure \ref{fig:localIN}). Thus, there is no dominating set $S$ such that $|\mathcal P_S|\geq |S|-1$ and Theorem \ref{th:aline} cannot be applied.

\end{document}